\documentclass[11pt,reqno]{amsart}
\usepackage{graphicx}

\usepackage{amscd,amssymb,amsmath,amsthm}
\usepackage{graphicx}
\usepackage{color}
\usepackage{cite}
\newtheorem{thm}{Theorem}

\newtheorem{lemma}{Lemma}
\newtheorem{pro}{Proposition}
\newtheorem{rk}{Remark}

\numberwithin{equation}{section} 
\def \t {\theta}

\def\s{\sigma}

\def\s{\sigma}

\def\s{\sigma}

\def\t{\theta}

\def \t {\theta}

\begin{document}
\title[Boundary conditions for translation-invariant Gibbs measures]{Boundary conditions for translation-invariant Gibbs measures of the Potts model on Cayley trees}

\author{D. Gandolfo, M.M. Rahmatullaev,  U. A. Rozikov}

 \address{D.\ Gandolfo \\Centre de Physique Th\'eorique, UMR 7332,
 Aix Marseille Univ, Universit\'e de Toulon, CNRS, CPT, Marseille, France.}
\email {gandolfo@cpt.univ-mrs.fr}

\address{M. \ M. \ Rahmatullaev and U.\ A.\ Rozikov\\ Institute of mathematics,
29, Do'rmon Yo'li str., 100125, Tashkent, Uzbekistan.}
\email {mrahmatullaev@rambler.ru\ \   rozikovu@yandex.ru}
\begin{abstract}
We consider translation-invariant splitting Gibbs measures (TISGMs) for the $q$-state Potts model on a Cayley tree of order two.
Recently a full description of the TISGMs was obtained, and it was shown
 in particular that at sufficiently low temperatures their number is  $2^{q}-1$.
 In this paper  for each TISGM $\mu$
 we explicitly give the set of boundary conditions such that limiting Gibbs measures with respect to these boundary conditions coincide with $\mu$.
 \end{abstract}
\maketitle



\section{Introduction.}

The analysis of translational invariant splitting Gibbs measures of the $q$-state Potts model on Cayley trees is based on the classification
 of translation-invariant boundary laws  which are in one-to-one correspondence with the TISGMs.
Recall that boundary laws are length-$q$ vectors which satisfy a
non-linear fixed-point equation (tree recursion).

It has been known for a long time that for
the anti-ferromagnetic Potts model there exists a unique TISGM \cite{Ro9}
and for the ferromagnetic Potts model
at sufficiently low temperatures there are at least $q+1$
translation-invariant Gibbs measures
\cite{Ga8}, \cite{Ga9}.

One of the $q+1$ well-known measures mentioned above is
obtained as infinite-volume limit of
the finite-dimensional Gibbs measures with free boundary condition and each of the remaining
$q$ measures is obtained as the corresponding limit with the boundary conditions of homogeneous
(constant) spin-configurations. While the $q$ measures
with homogeneous boundary conditions are always extremal in the set of all Gibbs measures \cite{Ga8}, \cite{Ga9},
for the free boundary condition measure there is a temperature, denoted by $T_0$, which is below the transition temperature, such that the measure is an extremal Gibbs measure if $T\geq T_0$ and loses its extremality for even lower temperatures \cite[Theorem 5.6.]{Ro}.

Recently, in \cite{KRK}  all TISGMs for the Potts model were found on the Cayley tree
of order $k\geq 2$, and it is shown that at sufficiently low
temperatures their number is  $2^{q}-1$. We note that the number of TISGMs does not depend on $k\geq 2$.
 In the case $k=2$ the explicit formulae for the
critical temperatures and all TISGMs are given.

In \cite{KR} some regions for the temperature ensuring that a given TISGM is (non-)extreme
 in the set of all Gibbs measures are found.
 In particular the existence of a temperature interval is shown for which there are at least
$2^{q-1} + q$ extremal TISGMs.

The fact that these measures can never be nontrivial convex combinations of each other
(i.e., they are extremal in the set of all TIGMs) is almost automatic (see  \cite[Theorem 2]{KRK}). However it is not clear what kind of boundary
conditions  are needed to get the remaining $2^q-q-2$ TISGMs as corresponding limits with the boundary conditions.
In this paper we shall answer this question. It is non-trivial problem since the number of TISGMs (i.e., $2^q-1$) is
larger than the number (i.e., $q$) of translation-invariant configurations. Therefore one expects to need
non-translation-invariant boundary conditions for some TISGMs. Concerning the Ising model,
the dependence of TISGMs on boundary conditions has been studied in \cite{YH}.

The paper is organized as follows. Section 2 contains preliminaries (necessary definitions
and facts). In section 3 we will show how to connect boundary laws with
boundary conditions, moreover we shall give the list of known TISGMs.
Section 4 contains our main result, namely given any  TISGM $\mu$, we show how to compute explicitly
a set of boundary conditions such that the limiting Gibbs measures with respect to these boundary conditions coincide with $\mu$. In the last section we construct concrete boundary conditions.

\section{Definitions}

Let $\Gamma^k= (V , L)$ be the regular Cayley tree, where each vertex has
$k + 1$ neighbors with $V$ being the set of vertices and $L$ the set of edges.

Two vertices  $t,s\in V,
(t\neq s)$ are called {\it neighbors} if they are connected by an edge. In this case we write
$\langle s, t \rangle$. Each vertex of $\Gamma^k$ has $k+1$ neighbors.

Fix an origin 0 of  $\Gamma^k$. We write $s\rightarrow t$,
if $t\neq s$ and the path connecting $0$ and $t$ passes  through $s$.
If $s\rightarrow t$ and $s$, $t$ are neighbors, then $t$
is called a {\it direct successor} of $s$ and this we write
as $s\rightarrow_{1} t$.

For any finite $A\subset V$, the boundary $\partial A$ of $A$ is
$$
\partial A=\{t\in V\setminus A: \exists x\in A, \langle
x,t\rangle\}.
$$

For every $A\subset V$, let $\Omega_A=\{1,2,...,q\}^A$ be the set of all possible spin
configurations on $A$. For brevity we write $\Omega$ instead of $\Omega_V$.

For every $A\subset V$ we define the $\sigma$-algebra $\textit{B}_A$ by
$$\textit{B}_A=\mbox{the} \, \sigma-\mbox{algebra generated by} \ \ \{X_t, t\in A\},$$
where
$X_t(\s)=\s(t)$ for all $t\in A, \s \in \Omega$. For brevity we write
$\textit{B}$ instead of  $\textit{B}_V$.

Let $A$ be a finite subset of $V$, $\omega\in\Omega$ and
$\s\in\Omega_A$. We define Potts interaction energy on $A$ given the inner
configuration $\s$ and the boundary condition $\omega$
by
\begin{equation}\label{1}
E_A^\omega(\s)=-J\sum_{{t,s\in A:\atop \langle
t,s\rangle}
}\delta_{\s(t)\s(s)}-J\sum_{{t\in A,\, s\in\partial A: \atop \langle
t,s\rangle} }\delta_{\s(t)\omega(s)},
\end{equation}
where $J\in \mathbb R$ and $\delta$ is the Kronecker's delta. 

A finite Gibbs measure $P_A^{\omega}$ on $\Omega_A$ corresponding to $E^{\omega}_A$ is defined by
\begin{equation}\label{2}
P_A^\omega(\s)=[Z_A^\omega]^{-1}\exp[-E_A^\omega(\s)], \s\in \Omega_A,
\end{equation}
where
$Z_A^\omega=\sum_{\widehat{\s}\in\Omega_A}\exp[-E_A^\omega(\widehat{\s})]$.
As usual $P_A^{\omega}$ can be considered as a probability measure on
$(\Omega, \textit{B})$.

For fixed $J$, if there is an increasing sequence of finite subsets $\{V_n\}$ such that
$V_n\nearrow V$ as $n\rightarrow\infty$ and
$P^\omega=\mbox{w}-\lim_{n\rightarrow\infty}P^\omega_{V_n}$
(the weak convergence of measures) exists for suitable fixed
$\omega\in\Omega$, then $P^\omega$ is called a \textit{limiting Gibbs measure}
with boundary condition $\omega$ for $J$. On the other hand, a Gibbs measure $P$ for $J$
is defined as a probability measure on  $(\Omega, \textit{B})$
such that for every $M$ in $\textit{B}_A$
\begin{equation}\label{3}
P(M|\textit{B}_{A^c})(\omega)=P_A^\omega(M). \ \ \ \ \ a.s.(P)
\end{equation}

It is known (\cite{Ge}, \cite{YH}) that the set $\Im(J)$ of all Gibbs measures for a fixed $J$ is a non-empty, compact convex set.
A limiting Gibbs measure is a Gibbs measure for the same $J$. Conversely, every extremal point of $\Im(J)$
is a limiting Gibbs measure
with a suitable boundary condition for the same $J$.
It is known (see page 241 of \cite{Ge}) that any extreme Gibbs measure of a Hamiltonian with nearest-neighbor interactions is a {\it splitting} Gibbs measure (which is equivalently called a tree-indexed Markov chain \cite{Ge}). Consequently,
any non-splitting Gibbs measure is not extreme. However, any splitting Gibbs measure (not necessary extreme) is a limiting Gibbs measure, because  it corresponds to a (generalized)\footnote{Adding a boundary field at each site of the boundary is called a generalized boundary condition \cite{GRRR} or boundary law \cite{Ge}}  boundary condition satisfying a compatibility (tree recursion) condition of Kolmogorov's theorem.
 In \cite{C} it was shown that for non-extremal Gibbs measures on $\mathbb Z^d$
a Gibbs measure need not be a limiting Gibbs measure (see \cite{C} and \cite{FV} for more details).

\section{Translation-invariant limiting Gibbs measures}

Let $|t|$ denote the distance between $0$ and  $t\in V$, i.e.
$|t|=n$ if there exists a chain
$0\rightarrow_1u_1\rightarrow_1u_2\rightarrow_1u_3\rightarrow_1...\rightarrow_1u_{n-1}\rightarrow_1t.$
We only consider the sequence of boxes
$$V_n=\{t\in V:|t|\leq n\}, n\geq 1.$$

For every $s\in V$ we define
$$\Gamma^k_s=\{s\}\cup\{t\in V:
s\rightarrow t\}, \ \ \mbox{and} \ \ V_{n,s}=\Gamma^k_s\cap V_n, \, n\geq 1.$$

Elements of the set $V_1\setminus\{0\}$ are the nearest neighbors of the origin $0$, 
since this set contains $k+1$ elements, we number them by $1, 2, \dots, k+1.$
We note that the subtrees $V_{n,i}$, $i=1,2,\dots,k+1$ are similar to each other, i.e.,
for any $i,j\in \{1,\dots,k+1\}$ the subtree $V_{n,i}$  can be obtained from the $V_{n,j}$ by a rotation around 0.
Moreover we have
\begin{equation}\label{Vv}
V_n=\{0\}\cup\bigcup_{i=1}^{k+1}V_{n,i}, \ \ \mbox{and} \ \  V_{n,i}=\{i\}\cup \bigcup_{j: i\rightarrow_{1} j}V_{n-1,j}.
\end{equation}
Using (\ref{Vv}),  from (\ref{1}) we get
\begin{equation}\label{1e}
E_{V_n}^\omega(\s)=\sum_{i=1}^{k+1}\left(E_{V_{n,i}}^\omega(\s)-J\delta_{\s(0)\s(i)}\right).
\end{equation}

For every $\omega\in \Omega, s\in V\setminus \{0\}$ and $n\geq|s|$ define
\begin{equation}\label{4}
W_{n,s}^\omega(l)=\sum_{\s \in \Omega_{V_{n,s}}:\, \s(s)=l}\exp[-E_{V_{n,s}}^\omega(\s)-J\delta_{l\omega(t)}], \ \
l=1,2,...,q,
\end{equation}
\begin{equation}\label{5}
R_{n,s}^l(\omega)=\frac{W_{n,s}^\omega(l)}{W_{n,s}^\omega(q)}, \ \ l=1,2,...,q,
\end{equation}
here $t$ is the unique vertex such that $t\rightarrow_1s$.

Now by (\ref{2}) and (\ref{1e})-(\ref{5}) we get
$$\frac{P_{V_n}^\omega(\s(0)=l)}{P_{V_n}^\omega(\s(0)=q)}={\sum_{\s:\s(0)=l}\exp\left(-E_{V_n}^\omega(\s)\right)\over
\sum_{\s:\s(0)=q}\exp\left(-E_{V_n}^\omega(\s)\right)}
$$
\begin{equation}\label{6}
=\prod_{i=1}^{k+1}\frac{(\exp(J)-1)R_{n,i}^l(\omega)+\sum_{p=1}^{q-1}R_{n,i}^p(\omega)+1}
{\exp(J)+\sum_{p=1}^{q-1}R_{n,i}^p(\omega)}, \ \ l=1,2,...,q,
\end{equation}
By (\ref{Vv}) and (\ref{4}) we obtain
\begin{equation}\label{7}
W_{n,s}^\omega(l)=\prod_{u:s\rightarrow_1u}[(\exp(J)-1)W_{n,u}^\omega(l)+\sum_{p=1}^q
W_{n,u}^\omega(p)], \ \ l=1,2,...,q,
\end{equation}
and for $n>m,\ \ \eta\in\Omega_{V_m},$ we get
\begin{equation}\label{8}
P_{V_n}^\omega(\{\s(s)=\eta(s), s\in
V_m\})=\frac{\exp[-E_{V_{m-1}}^\eta(\eta)]\prod_{s\in\partial
V_{m-1}}W_{n,s}^\omega(\eta(s))}{\sum_{\xi\in\Omega_{V_m}}\exp[-E_{V_{m-1}}^\xi(\xi)]\prod_{s\in\partial
V_{m-1}}W_{n,s}^\omega(\xi(s))}.
\end{equation}

From the above equalities we obtain the following

\begin{lemma}\label{l1} Let $\omega \in \Omega$ be given. If there is
$N>0$ such that $R_{n,s}^l(\omega)$ converges as
$n\rightarrow\infty$ for every $s\in V\setminus V_N$ and for every $l=1,...,q$, then
$P^\omega=\mbox{w}-\lim_{n\rightarrow\infty}P^\omega_{V_n}$ exists.
\end{lemma}

  For $n\geq 1$, $p=1,2,...q$, $i=1,...,k+1$ we denote
$$A_n=\{t\in V: |t|=n\},  \ \ \textit{N}^{(p)}_{n,i}(\sigma)=|\{x\in A_n\cap
V_{n,i}:\sigma(x)=p\}|.$$

\begin{lemma}\label{l2} Let $l=1,\dots,q$ and $\omega$ be a configuration such that\footnote{The sum in the RHS of $c^l(\omega)$ is taken over all {\it direct successors} of $t$, it should not be confused with the sum over all neighbors of $t$, i.e., $\sum_{s:\langle t,s\rangle }$.}
$$c^l(\omega)=\sum_{s:t\rightarrow_1s}\delta_{l\omega(s)}$$ is independent of $t\in V\setminus\{0\}.$
Then $R^l_{n,i}(\omega)=R^l_{n,j}(\omega)$ for any $i,j=1,2,...,k+1.$
\end{lemma}
\begin{proof} Since
$R_{n,i}^l(\omega)=\frac{W_{n,i}^\omega(l)}{W_{n,i}^\omega(q)},$
it  suffices to prove that
$W_{n,i}^\omega(l)=W_{n,j}^\omega(l)$ for any
$i,j=1,2,...,k+1.$

For the Hamiltonian we have

\begin{equation}\label{*}
E^{\omega}_{V_{n,i}}(\sigma)=E^{\sigma}_{V_{n-1,i}}(\sigma)-J\sum_{x\in
A_n\cap V_{n,i}}\sum_{x\rightarrow_1
y}\delta_{\sigma(x)\omega(y)}.
\end{equation}

By the condition of lemma \ref{l2} we obtain
$$
E^{\omega}_{V_{n,i}}(\sigma)=E^{\sigma}_{V_{n-1,i}}(\sigma)-J\sum_{x\in
A_n\cap
V_{n,i}}c^{\sigma(x)}(\omega)=E^{\sigma}_{V_{n-1,i}}(\sigma)-J\sum_{p=1}^q
\textit{N}^{(p)}_{n,i}(\sigma)c^{p}(\omega).
$$

For  $W_{n,i}^\omega(l)$ we have
$$
W_{n,i}^\omega(l)=\sum_{\sigma\in
V_{n,i}:\sigma(i)=l}\exp[E^{\sigma}_{V_{n-1,i}}(\sigma)-J\sum_{p=1}^q
\textit{N}^{(p)}_{n,i}(\sigma)c^{p}(\omega)-J\delta_{l\omega(t)}].
$$

Since $V_{n,i}$ is similar to $V_{n,j}$, for any $i,j\in \{1,\dots,k+1\}$, $n\geq 1$, there is an one-to-one correspondence $\gamma$ between
sets $\Omega_{V_{n,i}}$ and $\Omega_{V_{n,j}}$, which can be obtained by a rotation of the $V_{n,i}$ on the set
$V_{n,j}$. We note that the Potts interaction energy (\ref{1}) is
translation-invariant and by the condition of the lemma the quantity
$c^l(\omega)$ also does not depend on vertices of the tree. Therefore, if $\gamma(\sigma)=\varphi$ then
$$E^{\sigma}_{V_{n-1,i}}(\sigma)=E^{\gamma(\sigma)}_{V_{n-1,j}}(\gamma(\sigma))=E^{\varphi}_{V_{n-1,j}}(\varphi),$$
$$\textit{N}^{(p)}_{n,i}(\sigma)=\textit{N}^{(p)}_{n,j}(\varphi), \ \ \forall i,j=1,\dots,k+1, \, p=1,\dots,q.$$
Using these equalities we get
$$
W_{n,i}^\omega(l)=\sum_{\sigma\in
\Omega_{V_{n,i}}:\sigma(i)=l}\exp[E^{\sigma}_{V_{n-1,i}}(\sigma)-J\sum_{p=1}^q
\textit{N}^{(p)}_{n,i}(\sigma)c^{p}(\omega)-J\delta_{l\omega(t)}]=
$$ $$
\sum_{\varphi\in
\Omega_{V_{n,j}}:\varphi(j)=l}\exp[E^{\varphi}_{V_{n-1,j}}(\varphi)-J\sum_{p=1}^q
\textit{N}^{(p)}_{n,j}(\varphi)c^{p}(\omega)-J\delta_{l\omega(t)}]=W_{n,j}^\omega(l).
$$
Thus $R^l_{n,i}(\omega)=R^l_{n,j}(\omega)$  for any $i,j=1,2,...,k+1.$
\end{proof}
From the above proof it follows that $R_{n,s}^l(\omega)$ depends only on $n-|s|$, i.e., we have
\begin{equation}\label{9'}
R_{n,s}^l(\omega)=R_{n-|s|+1}^l(\omega), \ \ l=1,2,...q.
\end{equation}
Then from (\ref{7}) we get the following
\begin{equation}\label{9}
Y_n^l(\omega)=kF_l(Y_{n-1}^1(\omega),
Y_{n-1}^2(\omega),...,Y_{n-1}^{q-1}(\omega)),
\end{equation}
where $l=1,...,q-1$, $n\geq 2$, $Y_n^l(\omega)=\ln R_n^l(\omega)$ and
$F=(F_1,...,F_{q-1})$ with coordinates
\begin{equation}\label{10}
F_l(x^1, x^2,...,x^{q-1})=\ln
\frac{(\exp{(J)}-1)\exp{(x^l)}+\sum_{p=1}^{q-1}\exp{(x^p)}+1}{\exp{(J)}+\sum_{p=1}^{q-1}\exp{(x^p)}}.
\end{equation}

It is clear that if $Y^i(\omega)$ is a limit point of $Y^i_n(\omega)$, as $n\to\infty$ then by (\ref{9}) we get
\begin{equation}\label{11}
Y^l(\omega)=kF_l(Y^1(\omega), Y^2(\omega),...,Y^{q-1}(\omega)), \ \ l=1,2,..,q-1.
\end{equation}

For convenience we denote
$$\t=\exp(J), \ \ h_l=Y^l(\omega), \ \ l=1,2,...q-1.$$
Then the system (\ref{11}) becomes
\begin{equation}\label{pt}
h_i=k\ln\left({(\theta-1)e^{h_i}+\sum_{j=1}^{q-1}e^{h_j}+1\over
\theta+ \sum_{j=1}^{q-1}e^{h_j}}\right),\ \ i=1,\dots,q-1.
\end{equation}

In \cite{KRK} it is proven that to each solution of (\ref{pt})
corresponds a unique Gibbs measure which is called a translation-invariant splitting Gibbs measure (TISGM).

In \cite{KRK} all solutions of the equation (\ref{pt}) are given.
By these solutions the full set of TISGMs is described.
In particular, it is shown that any TISGM of the Potts model corresponds
to a solution of the following equation
\begin{equation}\label{rm}
h=f_m(h)\equiv k\ln\left({(\theta+m-1)e^h+q-m\over me^h+q-m-1+\theta}\right),
\end{equation}
for some $m=1,\dots,q-1$.

Denote
\begin{equation}\label{tm}
\theta_m=1+2\sqrt{m(q-m)}, \ \ m=1,\dots,q-1.
\end{equation}
It is easy to see that
\begin{equation}\label{st}
\theta_m=\theta_{q-m} \ \ \mbox{and} \ \ \theta_1<\theta_2<\dots<\theta_{\lfloor{q\over 2}\rfloor-1}<\theta_{\lfloor{q\over 2}\rfloor}\leq q+1.
\end{equation}

\begin{pro}\label{pw}\cite{KRK} Let $k=2$, $J>0$.
\begin{itemize}
\item[1.]
If $\theta<\theta_1$ then there exists a unique TISGM;
\item[2.]
If $\theta_{m}<\theta<\theta_{m+1}$ for some $m=1,\dots,\lfloor{q\over 2}\rfloor-1$ then there are  $1+2\sum_{s=1}^m{q\choose s}$ TISGMs which correspond  to the solutions $h_i\equiv h_i(\theta, s)=2\ln[x_i(s,\theta)]$, $i=1,2$  $s=1,\dots,m$ of (\ref{rm}), where
\begin{equation}\label{s}\begin{array}{ll}
x_{1}(s,\theta)={\theta-1-\sqrt{(\theta-1)^2-4s(q-s)}\over 2s},\\[2mm] x_{2}(s,\theta)={\theta-1+\sqrt{(\theta-1)^2-4s(q-s)}\over 2s}.
\end{array}
\end{equation}

\item[3.] If $\theta_{\lfloor{q\over 2}\rfloor}<\theta\ne q+1$ then there are $2^q-1$ TISGMs;

\item[4] If $\theta=q+1$ the number of TISGMs is as follows
$$\left\{\begin{array}{ll}
2^{q-1}, \ \ \mbox{if} \ \ q \ \ \mbox{is odd}\\[2mm]
2^{q-1}-{q-1\choose q/2}, \ \ \mbox{if} \ \ q \ \ \mbox{is even;}
\end{array}\right.$$

\item[5.] If $\theta=\theta_m$, $m=1,\dots,\lfloor{q\over 2}\rfloor$, \,($\theta_{\lfloor{q\over 2}\rfloor}\ne q+1$) then the number of TISGMs is
$$1+{q\choose m}+2\sum_{s=1}^{m-1}{q\choose s}.$$
\end{itemize}
\end{pro}

The number of TISGMs does not depend on $k\geq 2$ see \cite[Theorem 1]{KRK}. But for $k\geq 3$ explicit formulas for the solutions are not known.
Therefore in this paper we consider only the case $k=2$.

Following \cite{KRK} we note that each TISGM corresponds to a solution of (\ref{rm}) with some $m\leq \lfloor{q\over 2}\rfloor$.
Moreover, for a given $m\leq \lfloor{q\over 2}\rfloor$,
a fixed solution $h_i(\theta,m)$ to (\ref{rm}) generates ${q\choose m}$ vectors by permuting coordinates of the vector $(\underbrace{h_i,h_i,\dots,h_i}_m,\underbrace{0,0,\dots,0}_{q-m})$ and giving ${q\choose m}$ TISGMs.
Thus without loss of generality we can only consider
the measure $\mu_i(\theta,m)$ corresponding to vector ${\bf h}(m,i)=(\underbrace{h_i,h_i,\dots,h_i}_m,\underbrace{0,0,\dots,0}_{q-m-1})$, i.e., normalized on $q$th coordinate (see Remark 2 and Corollary 2 of \cite{KRK}).
Denote by $\mu_0\equiv\mu_0(\theta)$ the TISGM corresponding to solution $h_i\equiv 0$ and by $\mu_i\equiv \mu_i(\theta,m)$ the TISGM corresponding to the solution $h_i(\theta, m)$, $i=1,2$, $m=1,\dots, \lfloor{q\over 2}\rfloor$ (given in Proposition \ref{pw}).
In this paper our aim is to obtain measures $\mu_i$ by changing boundary conditions.

\section{Boundary conditions for TISGMs}

The following lemma can be proved by simple analysis.

\begin{lemma}\label{l3} \begin{itemize}
\item[i.] For $k\geq 2$ and $\theta>1$ the function $f_m(h)$, $h\in \mathbb R$ defined in (\ref{rm}) has the following properties:
\begin{itemize}

\item[a)] $\{h: f_m(h)=h\}=\{0, h_1, h_2\}$, if $\theta>\theta_m, m\leq[ q/2]$;

\item[b)] $a<f_m(h)<A, \ \ \mbox{with} \ \ a=k\ln{q-m\over q+\theta-m-1}, \ \ A=k\ln{\theta+m-1\over m}$;

\item[c)] ${d\over dh}f_m(h)={k(\theta-1)(\theta+q-1)e^h\over (me^h+\theta+q-m-1)((\theta+m-1)e^h+q-m)}>0$;
\end{itemize}

\item[ii.] If $k=2$ and $m\leq q/2$ then for solutions $h_1$ and $h_2$ mentioned in Proposition \ref{pw} the following statements hold
$$\begin{array}{llll}
0<h_1=h_2, \ \ \mbox{if} \ \ \theta=\theta_m\\[2mm]
0<h_1<h_2, \ \ \mbox{if} \ \ \theta_m<\theta<\theta_c, \ \ \mbox{with}\ \   \theta_c=q+1\\[2mm]
0=h_1<h_2, \ \ \mbox{if} \ \ \theta=q+1\\[2mm]
h_1<0<h_2, \ \ \mbox{if} \ \ q+1<\theta.
\end{array}$$
\end{itemize}
\end{lemma}

For each solution $h_i(\theta,m)$ we want to find
$\omega=\omega(h_i)\in \Omega$, such that
$\mu_i(\theta,m)=P^\omega$, where $P^\omega$ is defined in Lemma \ref{l1}.

Consider the dynamical system (\ref{9}) for $k=2$. Denote $G(h)=2F(h)$. For a given initial vector
$v^{(0)}=(v^{(0)}_1,\dots,v^{(0)}_{q-1})$, we shall study the limit
\begin{equation}\label{li}
\lim_{n\rightarrow \infty}G^{(n)}(v^{(0)}),
\end{equation}
here $G^{(n)}(v)=\underbrace{G(G(...G(v))...)}_n.$

Figures \ref{ff1}-\ref{ff4} show the streamlines of the vector field $G^{(n)}(v)$ for $k=2$, $q=3$.
These figures also illustrate the limit points of (\ref{li}).
\begin{figure}[h]
\centering
\includegraphics[width=9cm]{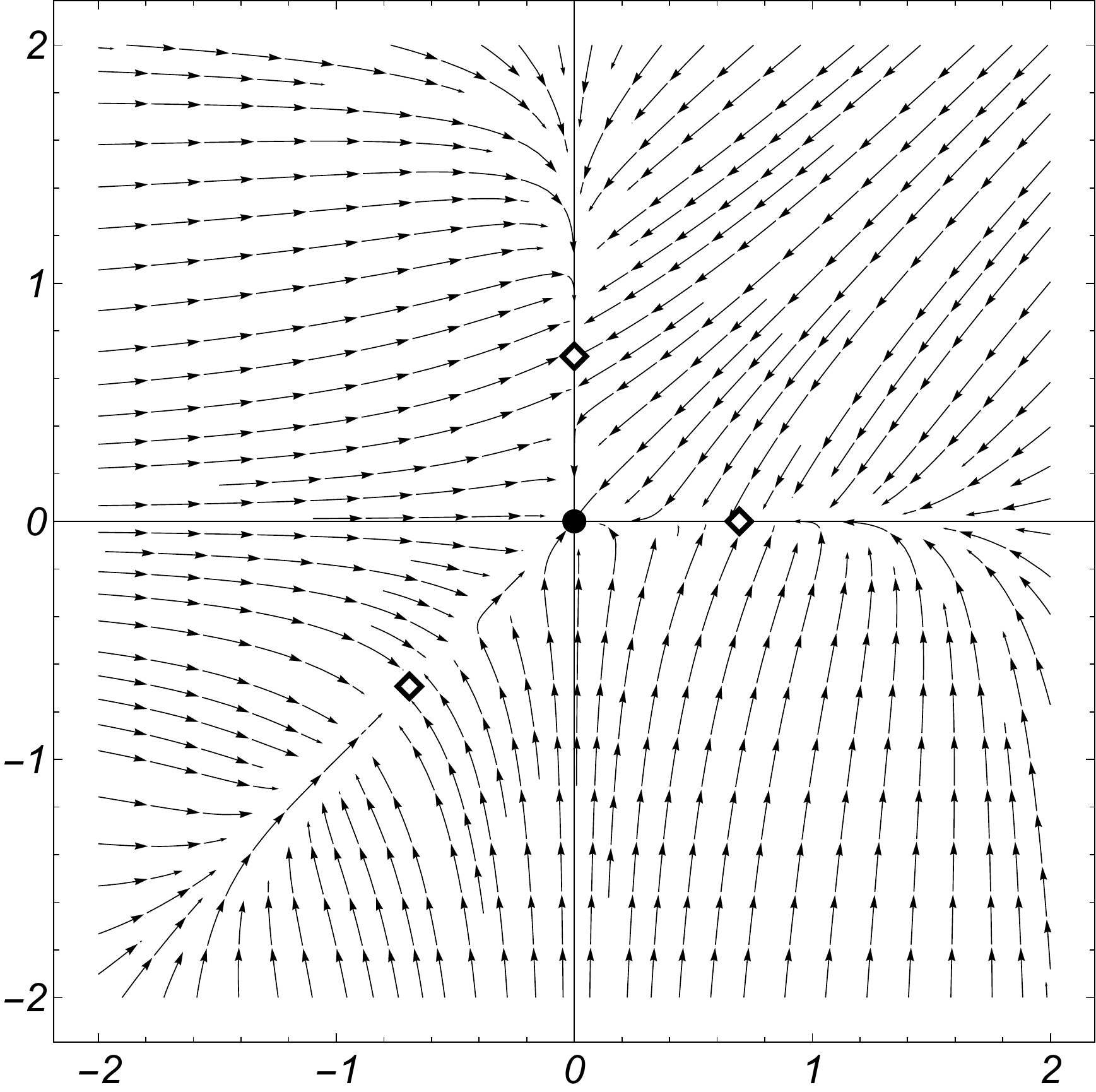}
\caption{Here and in Figures \ref{ff2}-\ref{ff4} the streamlines of the vector field $G^{(n)}(v)$ for $k=2$, $q=3$ are shown.
These figures also illustrate the limit points of (\ref{li}).  The plane is formed by a horizontal $v_1$-axis, and a vertical $v_2$-axis.  This figure applies for the case   $\theta=\theta_{1}=1+2\sqrt{2}$. Four fixed points. The origin is an attractor. There are 3 saddle fixed
points.}
\label{ff1}
\includegraphics[width=9cm]{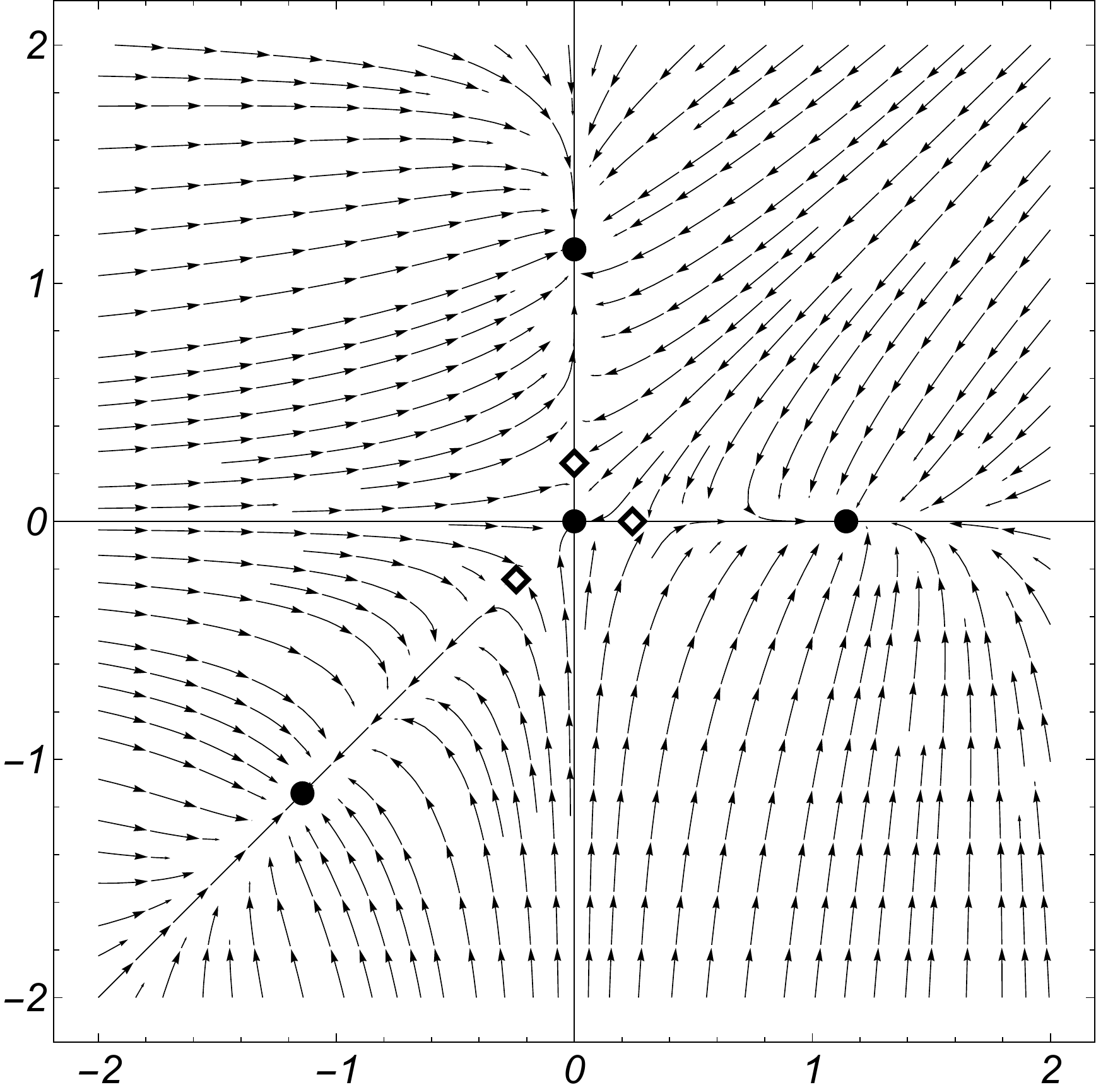}
\caption{$\theta=3.9>\theta_{1}$. Seven fixed points. Four of them (black dots) are attractors. Three (rectangular dots) points
are saddles coming from the saddle points of Fig.\ref{ff1}.}
\label{ff2}
\end{figure}
\begin{figure}
\centering
\includegraphics[width=9cm]{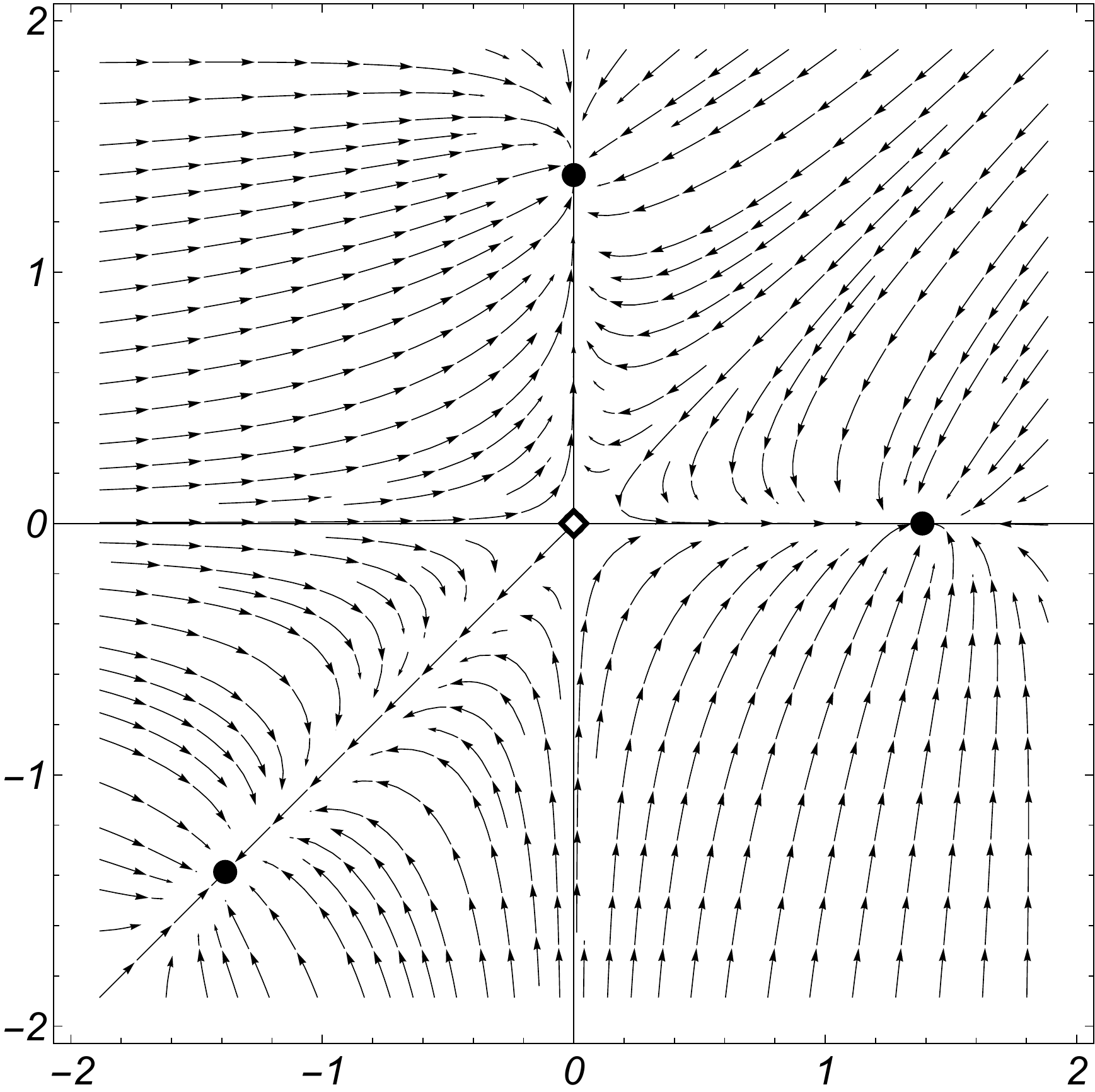}
\caption{$\theta=4$. Four fixed points. The origin is a repeller  point. Other fixed points are attractors.}
\label{ff3}
\includegraphics[width=9cm]{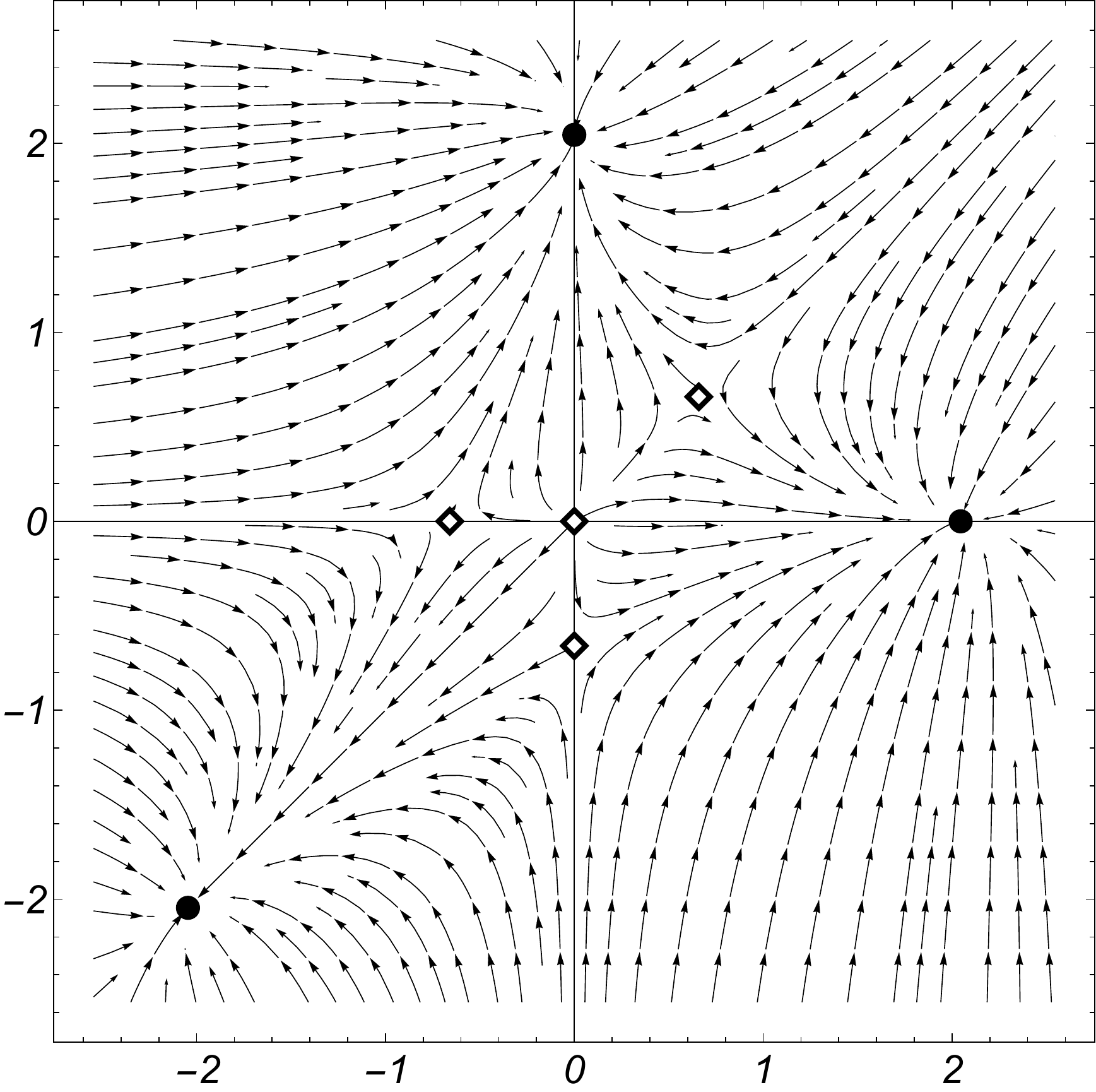}
\caption{$\theta=4.5$. Seven fixed points. The origin is a repeller,  other rectangular  dots
are saddles. The black dots are attractors.}
\label{ff4}
\end{figure}

Denote
$$I_m=\{v\in \mathbb R^{q-1}: v_1=\dots=v_m, \ \ v_{m+1}=\dots=v_{q-1}=0\}.$$
It is easy to see that the set $I_m$ is invariant with respect to $G$, i.e. $G(I_m)\subset I_m$.

The following lemma gives the limits of (\ref{li}) on the invariant $I_m$ (compare with Fig.\ref{ff1}-\ref{ff4}).
\begin{lemma}\label{l4} 1) If $\theta=\theta_m$, for some $m=1,\dots,\lfloor{q\over 2}\rfloor$ then
\begin{equation}\label{14a}
\lim_{n\rightarrow \infty}G^{(n)}(v^{(0)})=\left\{%
\begin{array}{ll}
    {\bf h}(m,1), & \hbox{\textit{if}}\ \ v^{(0)}\in I_m \ \ \mbox{and} \ \ v^{(0)}_1\geq h_1\\[2mm]
    (0,\dots, 0), & \hbox{\textit{if}}\ \ v^{(0)}\in I_m \ \ \mbox{and} \ \ v^{(0)}_1<h_1
    \end{array}%
\right.
\end{equation}

2) If $\theta_m<\theta<\theta_c=q+1$ then

\begin{equation}\label{14b}
\lim_{n\rightarrow \infty}G^{(n)}(v^{(0)})=\left\{%
\begin{array}{ll}
    {\bf h}(m,2), & \hbox{\textit{if}}\ \ v^{(0)}\in I_m \ \ \mbox{and} \ \ v^{(0)}_1> h_1\\[2mm]
    {\bf h}(m,1), & \hbox{\textit{if}}\ \ v^{(0)}\in I_m \ \ \mbox{and} \ \ v^{(0)}_1= h_1\\[2mm]
    (0,\dots, 0), & \hbox{\textit{if}}\ \ v^{(0)}\in I_m \ \ \mbox{and} \ \ v^{(0)}_1<h_1
\end{array}%
\right.
\end{equation}

3) If $\theta=\theta_c$ then

\begin{equation}\label{14ñ}
\lim_{n\rightarrow \infty}G^{(n)}(v^{(0)})=\left\{%
\begin{array}{ll}
   {\bf h}(m,2), & \hbox{\textit{if}}\ \ v^{(0)}\in I_m \ \ \mbox{and} \ \ v^{(0)}_1>0\\[2mm]
    (0,\dots, 0), & \hbox{\textit{if}}\ \ v^{(0)}\in I_m \ \ \mbox{and} \ \ v^{(0)}_1\leq 0
\end{array}%
\right.
\end{equation}

4) If $\theta>\theta_c$ then
\begin{equation}\label{14d}
\lim_{n\rightarrow \infty}G^{(n)}(v^{(0)})=\left\{%
\begin{array}{ll}
    {\bf h}(m,2), & \hbox{\textit{if}}\ \ v^{(0)}\in I_m \ \ \mbox{and} \ \ v^{(0)}_1>0\\[2mm]
    {\bf h}(m,1), & \hbox{\textit{if}}\ \ v^{(0)}\in I_m \ \ \mbox{and} \ \ v^{(0)}_1<0\\[2mm]
    (0,\dots, 0), & \hbox{\textit{if}}\ \ v^{(0)}\in I_m \ \ \mbox{and} \ \ v^{(0)}_1=0
\end{array}%
\right.
\end{equation}
\end{lemma}

\begin{proof} Restrict function $G(h)$ to $I_m$, then we get the $j$th coordinate of $G(h)$ (for any $j=1,\dots,m$) is equal to $f_m(h)$ which is introduced in (\ref{rm}). Other coordinates of $G(h)$ are equal to 0. By Lemma \ref{l3} we have that $f_m$ is an increasing function. Here we consider the case when the function $f_m$ has
  three fixed points $0, h_1,h_2$. This proof is more simple for cases when $f_m$ has two fixed points.  We prove the part 2), other parts are similar.
  In case 2), by Lemma \ref{l3} we have that $0<h_1<h_2$ and the point $h_1$
is a repeller, i.e., $f_m'(h_1)>1$ and the points $0, h_2$ are attractive, i.e., $f_m'(0)<1$, $f_m'(h_2)<1$.
Now we shall take arbitrary $x_0>0$ and prove that $x_n =f_m(x_{n-1})$, $n\geq 1$
converges as $n\to\infty$. Consider the following partition $(-\infty,+\infty) = (-\infty,0)\cup\{0\}\cup(0,h_1)\cup \{h_1\}\cup(h_1, h_2)\cup \{h_2\}\cup(h_2,+\infty)$. For any  $x\in (-\infty, 0)$ we have $x < f_m(x) <0$, since $f_m$ is an increasing function, from the last inequalities we get  $x<f_m(x) < f_m^2(x) < f_m(0)=0$. Iterating this argument we obtain $f_m^{n-1}(x)<f_m^{n}(x)<0$, which
for any $x_0\in (-\infty,0)$ gives $x_{n-1}<x_n<0$, i.e., $x_n$ converges and its limit is a fixed point of $f_m$, since $f_m$
has unique fixed point $0$ in $(-\infty, 0]$ we conclude that the limit is $0$. For $x\in (0, h_1)$
we have $h_1>x >f(x)>0$, consequently $x_n > x_{n+1}$, i.e., $x_n$ converges
and its limit is again $0$. Similarly, one can show that if $x_0>h_1$ then $x_n\to h_2$ as $n\to \infty$.
\end{proof}

By (\ref{9}) the asymptotic behavior of the vector $Y_n(\omega)=(Y_n^1(\omega), \dots, Y_n^{q-1}(\omega))$ depends only
on the vector $Y_1(\omega)$, where
\begin{equation}\label{nu}
Y_1^l(\omega)=J\left(c^l(\omega)-c^q(\omega)\right), \ \ l=1,\dots,q-1.
\end{equation}

For a given $m\in\{1,\dots,\lfloor{q\over 2}\rfloor\}$ and $J>0$ we introduce the following
sets of configurations:
$${\mathbb B}_{m}=\{\omega\in \Omega: c^1(\omega)=\dots =c^m(\omega),\ \ c^{m+1}(\omega)=\dots =c^{q-1}(\omega)=c^q(\omega)\},$$
$${\mathbb B}^+_{m,0}=\{\omega\in \mathbb B_m: c^1(\omega)>c^q(\omega)\},$$
$${\mathbb B}^0_{m,0}=\{\omega\in \mathbb B_m: c^1(\omega)=c^q(\omega)\},$$
$${\mathbb B}^-_{m,0}=\{\omega\in \mathbb B_m: c^1(\omega)<c^q(\omega)\},$$
 $${\mathbb B}^+_{m,1}=\{\omega\in \mathbb B_m: J\left(c^1(\omega)-c^q(\omega)\right)>h_1\},$$
  $${\mathbb B}^0_{m,1}=\{\omega\in \mathbb B_m: J\left(c^1(\omega)-c^q(\omega)\right)=h_1\},$$
$${\mathbb B}^-_{m,1}=\{\omega\in \mathbb B_m: J\left(c^1(\omega)-c^q(\omega)\right)<h_1\}.$$

Now taking the coordinates of an initial vector as in (\ref{nu}) by Lemma \ref{l1} and Lemma \ref{l4} we get the following

\begin{thm}\label{t2} 1) If $\theta=\theta_m$, for some $m=1,\dots,\lfloor{q\over 2}\rfloor$ then
\begin{equation}\label{1a}
P^\omega=\left\{%
\begin{array}{ll}
    \mu_1(\theta, m), & \hbox{\textit{if}}\ \ \omega\in \mathbb B_{m,1}^+\cup \mathbb B_{m,1}^0\\[2mm]
    \mu_0(\theta), & \hbox{\textit{if}}\ \ \omega\in \mathbb B^-_{m,1}
    \end{array}%
\right.
\end{equation}

2) If $\theta_m<\theta<\theta_c=q+1$ then

\begin{equation}\label{1b}
P^\omega=\left\{%
\begin{array}{lll}
    \mu_2(\theta, m), & \hbox{\textit{if}}\ \ \omega\in \mathbb B_{m,1}^+\\[2mm]
    \mu_1(\theta,m), & \hbox{\textit{if}}\ \ \omega\in \mathbb B^0_{m,1}\\[2mm]
     \mu_0(\theta), & \hbox{\textit{if}}\ \ \omega\in \mathbb B_{m,1}^-
\end{array}%
\right.
\end{equation}

3) If $\theta=\theta_c$ then

\begin{equation}\label{1ñ}
P^\omega=\left\{%
\begin{array}{ll}
    \mu_2(\theta, m), & \hbox{\textit{if}}\ \ \omega\in \mathbb B^+_{m,0}\\[2mm]
    \mu_0(\theta), & \hbox{\textit{if}}\ \ \omega\in \mathbb B^-_{m,0}\cup \mathbb B^0_{m,0}
\end{array}%
\right.
\end{equation}

4) If $\theta>\theta_c$ then
\begin{equation}\label{1d}
P^\omega=\left\{%
\begin{array}{lll}
    \mu_2(\theta, m), & \hbox{\textit{if}}\ \ \omega\in \mathbb B_{m,0}^+\\[2mm]
    \mu_1(\theta,m), & \hbox{\textit{if}}\ \ \omega\in \mathbb B^-_{m,0}\\[2mm]
     \mu_0(\theta), & \hbox{\textit{if}}\ \ \omega\in \mathbb B_{m,0}^0
\end{array}%
\right. .
\end{equation}
\end{thm}

In the next section we use Theorem \ref{t2} to construct some concrete boundary conditions.

\section{Construction of boundary conditions}

In this section for $k=2$, $J>0$, $q\geq 3$ and $m\in \{1,\dots,\lfloor{q\over 2}\rfloor\}$
we shall give examples of boundary configurations.

For $k=2$ by Proposition \ref{pw} there are up to $2^q-1$ TISGMs .
We shall consider only $\mu_0(\theta)$ corresponding to $h=(0,0,\dots,0)$ and
$\mu_i(\theta,m)$ corresponding to vector ${\bf h}(m,i)=(\underbrace{h_i,h_i,\dots,h_i}_m,\underbrace{0,0,\dots,0}_{q-m-1})$, $i=1,2$
with \begin{equation}\label{so}\begin{array}{ll}
h_1=2\ln{\theta-1-\sqrt{(\theta-1)^2-4m(q-m)}\over 2m},\\[2mm]
h_2=2\ln{\theta-1+\sqrt{(\theta-1)^2-4m(q-m)}\over 2m}.
\end{array}
\end{equation}
Using Theorem \ref{t2} we shall give some boundary conditions for each measure $\mu_i$. Boundary conditions for the remaining measures can be obtained by using the permutation symmetry of the Potts model.

{\bf Case $\mu_0$.} If $\theta<\theta_1$ then $\mu_0$ is a unique measure, and one can take any
boundary configuration $\omega$ to have $P^\omega=\mu_0$. But for $\theta\geq \theta_1$ one has to
check the conditions of Theorem \ref{t2} to have the limiting measure equal to $\mu_0$. 

For example, if $\theta=\theta_m<q+1$ for some $m\in \{1,\dots,\lfloor{q\over 2}\rfloor\}$ then we must take $\omega\in \mathbb B_{m,1}^-$, i.e.
\begin{equation}\label{w0}
J(c^1(\omega)-c^q(\omega))<h_1, \ \ c^1(\omega)=\dots=c^m(\omega), \ \ c^{m+1}(\omega)=\dots=c^q(\omega).
\end{equation}
\begin{rk} For a given TISGM to find its boundary condition one has to construct 
configurations $\omega$ which satisfy the system (like (\ref{w0})) 
derived by corresponding sufficient conditions of Theorem \ref{t2}. Below we give
several examples of such configurations.
It will be clear from our examples that some TISGM may have an infinite set of boundary configurations\footnote{This remark and some examples below are added corresponding to a suggestion of a reviewer.}.
\end{rk}

Since $h_1>0$ (see Lemma \ref{l3}), the system (\ref{w0}) is satisfied, {\it for
example}, if $\omega$ satisfies one of the following 
\begin{itemize}
\item $c^i(\omega)=0$, $i=1,2,\dots,m$, i.e., if $\omega(x)=i$, then on direct 
successors $x_1$, $x_2$ of $x$
one has $\omega(x_1)\ne i$, $\omega(x_2)\ne i$; 
and $c^j(\omega)=1$ for each $j=m+1,\dots,q$, i.e., 
if $\omega(x)=j$ then one has $\omega(x_1)=j$ but $\omega(x_2)\ne j$. 
See Fig.\ref{f1} for an example of  such
a configuration for $q=5$ and $m=2$.
\item $c^i(\omega)=1$, $i=1,2,\dots,m$, i.e.,
if $\omega(x)=i$ then $\omega(x_1)=i$ but $\omega(x_2)\ne i$ and it
contains $j\in \{m+1,\dots,q\}$ in such a way that if
$\omega(x)=j$ then on direct successors $x_1$, $x_2$ of $x$
one has $\omega(x_1)=\omega(x_2)=j$. In this case $c^i(\omega)=2$ for each $i=m+1,\dots,q$.
See Fig.\ref{f1a} for an example of such
a configuration for $q=15$ and $m=3$.
\end{itemize}
\begin{figure}
\centering
\includegraphics[width=12.5cm]{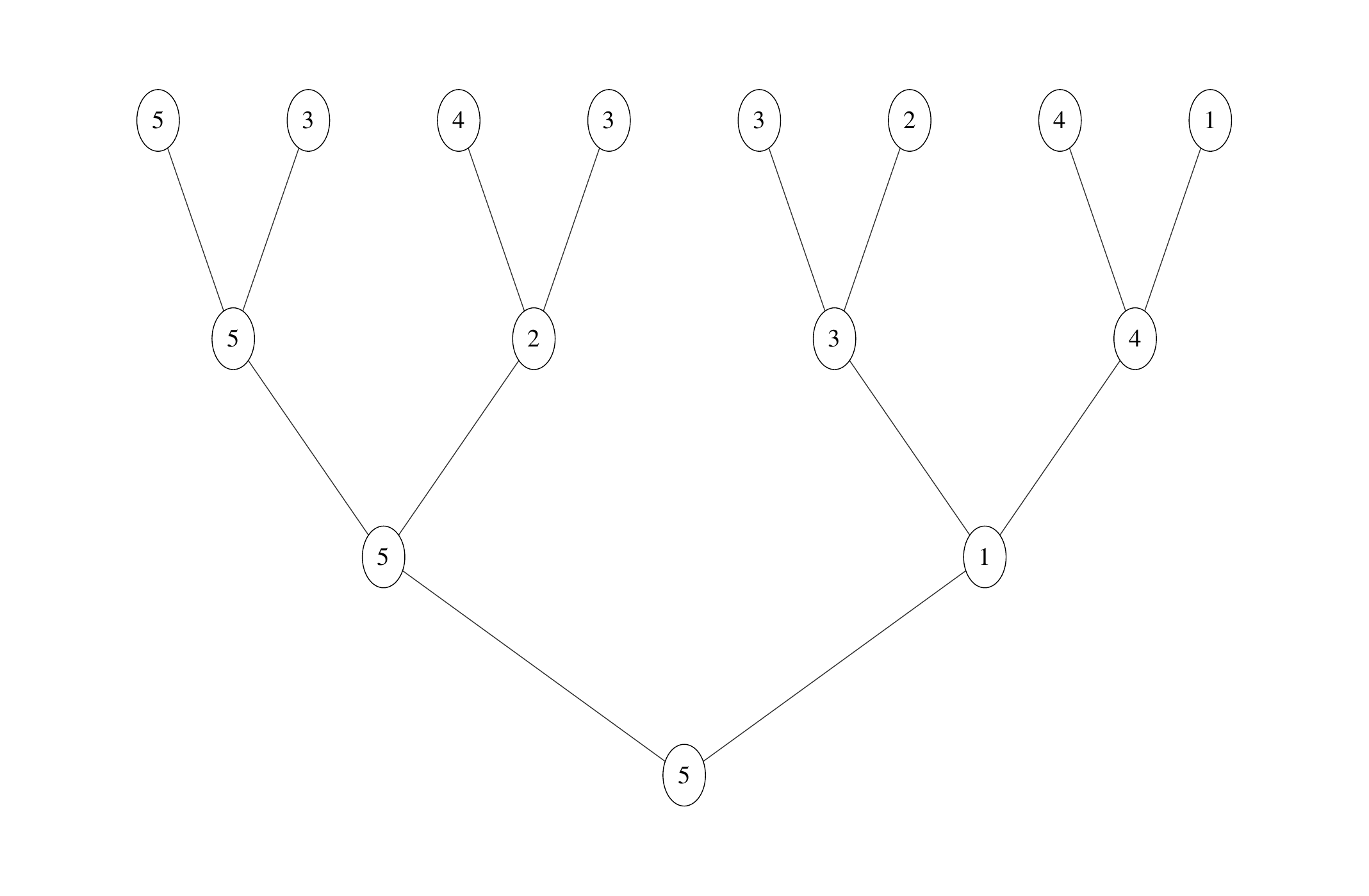}
\caption{An example of a boundary condition for the TISGM $\mu_0$, $q=5$, $m=2$. Here $c^1(\omega)=c^2(\omega)=0$, $c^3(\omega)=c^4(\omega)=c^5(\omega)=1$. Note that the values $1, 2$ occur sufficiently often keeping the conditions $c^1(\omega)=c^2(\omega)=0$.}
\label{f1}.
\end{figure}
\begin{figure}
\centering
\includegraphics[width=12.5cm]{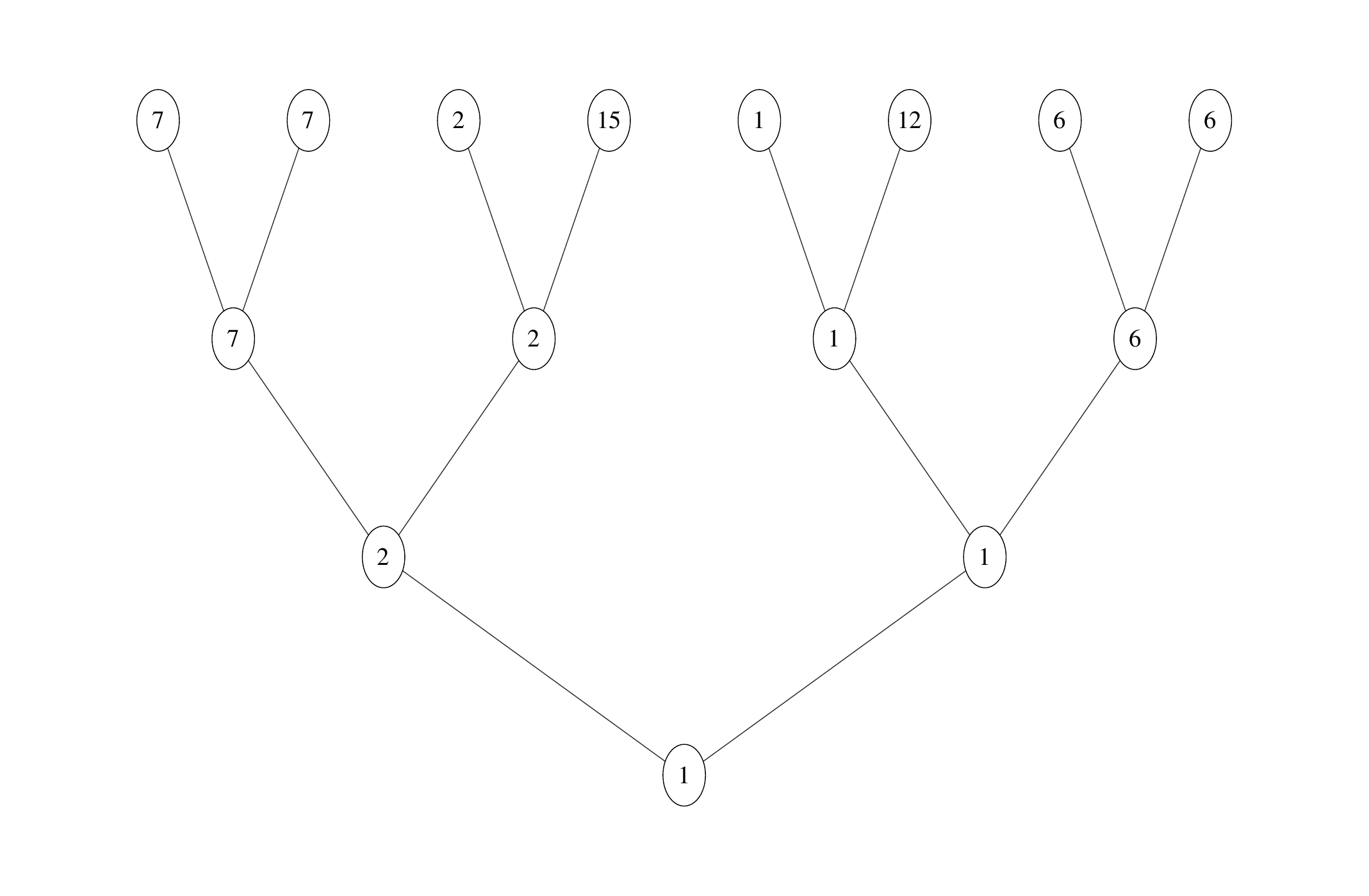}
\caption{An example of a boundary condition for the TISGM $\mu_0$, for $q=15$, $m=3$. Here we have $c^i(\omega)=1$, $i=1,2,3$ and $c^j(\omega)=2$, $j=4,5,\dots,15$.}
\label{f1a}.
\end{figure}
{\bf Case $\mu_1$.} This measure exists for $\theta\geq \theta_m$.

{\it Subcase $\theta=\theta_m$.} By Theorem \ref{t2} for $\mu_1$ we have the condition
$\omega\in \mathbb B_{m,1}^+\cup \mathbb B_{m,1}^0$, i.e.
\begin{equation}\label{m1}
J(c^1(\omega)-c^q(\omega))\geq h_1, \ \ c^1(\omega)=\dots=c^m(\omega), \ \ c^{m+1}(\omega)=\dots=c^q(\omega).
\end{equation}
Note that $J=\ln\theta$. Assume $\ln\theta_m\geq h_1=2\ln{\theta_m-1\over 2m}$ which is equivalent to the following
\begin{equation}\label{qm}
2m\sqrt{m^2+1}(\sqrt{m^2+1}-m)\leq q\leq 2m\sqrt{m^2+1}(\sqrt{m^2+1}+m).
\end{equation}
Under condition (\ref{qm}) the system (\ref{m1}) is satisfied {\it for
example}, if $m\geq 1$ and
$\omega$ is such that $c^i(\omega)=1$, $i=1,\dots,m$ and 
$c^j(\omega)=0$, $i=m+1,\dots,q$. See Fig.\ref{f2} for an example of such
a configuration for $q=5$ and $m=2$ and Fig.\ref{f2a} for configuration in case $q=10$, $m=4$. 

\begin{figure}
\centering
\includegraphics[width=12.5cm]{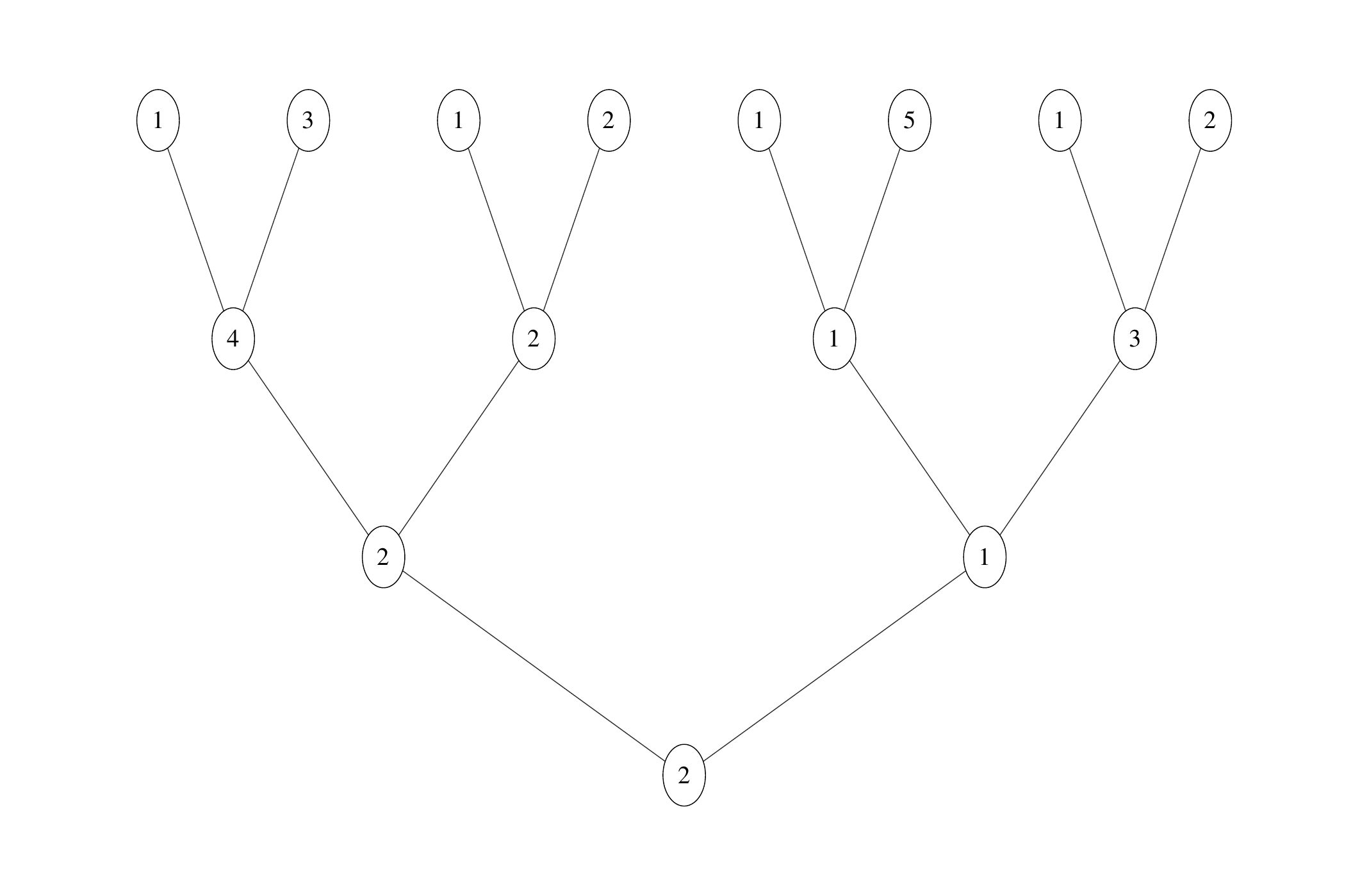}
\caption{An example of a boundary condition for the TISGM $\mu_1$, $q=5$, $m=2$. Here we have $c^1(\omega)=c^2(\omega)=1$, $c^3(\omega)=c^4(\omega)=c^5(\omega)=0$.}
\label{f2}
\end{figure}
\begin{figure}
\centering
\includegraphics[width=12.5cm]{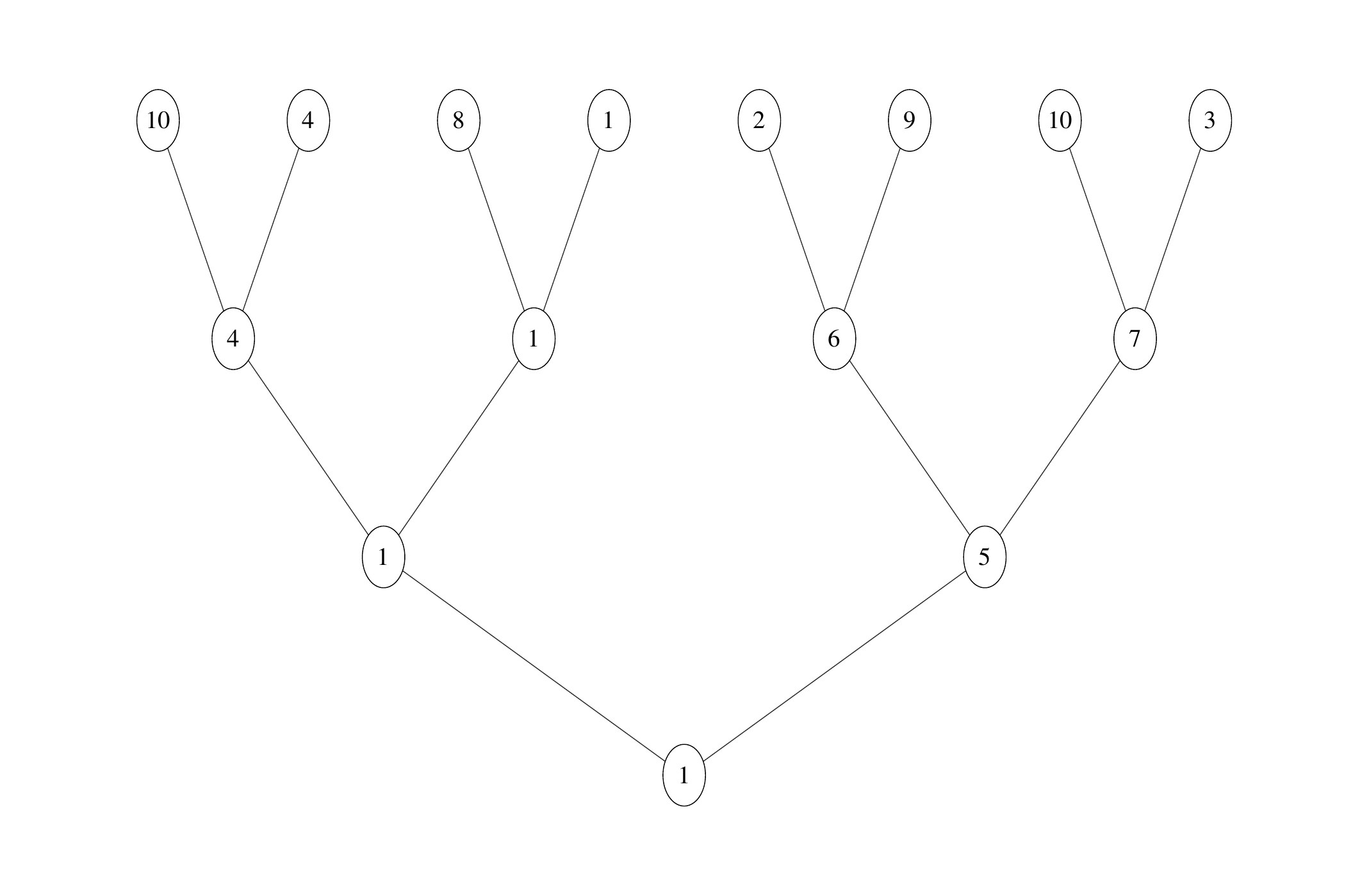}
\caption{An example of a boundary condition for the TISGM $\mu_1$, $q=10$, $m=4$. Here we have $c^i(\omega)=1$, $i=1,2,3,4$; $c^j(\omega)=0$, $j=5,\dots,10$.}
\label{f2a}
\end{figure}

{\it Subcase $\theta_m<\theta<q+1$.} From Theorem \ref{t2} for $\mu_1$ we have the condition
\begin{equation}\label{m11}
J(c^1(\omega)-c^q(\omega))=h_1, \ \ c^1(\omega)=\dots=c^m(\omega), \ \ c^{m+1}(\omega)=\dots=c^q(\omega).
\end{equation}
Assume $\theta$ is a solution to the equation $\ln\theta=h_1$.
Computer analysis shows that this equation has a solution if
for example $q=17$, $m=1$ or $q=55$, $m=2$.  So assuming existence
of such a solution $\theta=\theta^*$ one can take a boundary condition configuration as in the previous case (like in Fig.\ref{f2})

{\it Subcase $\theta=q+1$.} In this case we have $\mu_1=\mu_0$. Therefore the boundary condition can be taken as in Case $\mu_0$.

{\it Subcase $\theta>q+1$.} For $\mu_1$ we should have
\begin{equation}\label{m111}
c^1(\omega)-c^q(\omega)<0, \ \ c^1(\omega)=\dots=c^m(\omega), \ \ c^{m+1}(\omega)=\dots=c^q(\omega).
\end{equation}
we can take a configuration $\omega$ such that $c^i(\omega)=0$, $i=1,\dots,m$ and $c^j(\omega)=1$, $j=m+1,\dots,q$.
(See Fig.\ref{f1} for such a configuration).
\begin{rk} From above examples one can see that depending on the temperature (equivalently depending on the parameter $\theta$) 
a configuration may be the boundary condition for different TISGMs. For example, the configuration given in Fig.\ref{f1} is a boundary condition for TISGM $\mu_0$ if $\theta=\theta_2<q+1$, but the same configuration is the boundary condition for TISGM $\mu_1$ if $\theta>q+1$.   \end{rk}
{\bf Case $\mu_2$.} Check the conditions of Theorem \ref{t2}:

{\it Subcase $\theta=\theta_m$.} In this case we have $\mu_2=\mu_1$, i.e. the boundary condition is constructed in the previous case.

{\it Subcase $\theta_m<\theta<q+1$.} From Theorem \ref{t2} for $\mu_2$ we have the condition
\begin{equation}\label{m21}
J(c^1(\omega)-c^q(\omega))>h_1, \ \ c^1(\omega)=\dots=c^m(\omega), \ \ c^{m+1}(\omega)=\dots=c^q(\omega).
\end{equation}
If $\ln \theta>h_1$ then it is easy to see that $\omega$ satisfies the condition (\ref{m21}) if 
$m\geq 1$ and $\omega$  is such that $c^i(\omega)=1$, $i\in \{1,\dots,m\}$ and $c^j(\omega)=0$, $i\in \{m+1,\dots,q\}$ (like in Fig.\ref{f2})

{\it Subcase $\theta\geq q+1$.}  For $\mu_2$ we should have
\begin{equation}\label{m212}
c^1(\omega)-c^q(\omega)>0, \ \ c^1(\omega)=\dots=c^m(\omega), \ \ c^{m+1}(\omega)=\dots=c^q(\omega).
\end{equation}
Condition (\ref{m212}) is easily checkable. For example, configurations shown in Fig. \ref{f2} and Fig.\ref{f2a} satisfy this condition.   
\section*{ Acknowledgements}
U.Rozikov thanks Aix-Marseille University Institute for Advanced Study IM\'eRA
(Marseille, France) for support by a residency scheme.
We thank both referees for their helpful suggestions.

\end{document}